\documentclass[conference]{IEEEtran}
\IEEEoverridecommandlockouts

\usepackage{cite}
\usepackage{amsmath,amssymb,amsfonts}
\usepackage{algorithmic}
\usepackage{graphicx}
\usepackage{textcomp}
\usepackage{xcolor}
\def\BibTeX{{\rm B\kern-.05em{\sc i\kern-.025em b}\kern-.08em
    T\kern-.1667em\lower.7ex\hbox{E}\kern-.125emX}}

\newcommand{\ha}{{\mathcal{H}_\mathrm{1}}}
\newcommand{\hb}{{\mathcal{H}_\mathrm{2}}}

\newcommand{\rd}{\signal{\rho}_\mathrm{d}}
\newcommand{\gn}{\signal{g}_{\mathrm{d},n}}
\newcommand{\vd}{V_\mathrm{d}}
\newcommand{\kd}{\mathcal{K}_\mathrm{d}}

\newcommand{\hmr}{\mathcal{H}_{\signal{M}}}
\newcommand{\hmra}{\mathcal{H}_{\signal{M}_\alpha}}
\newcommand{\hsignal}[1]{\bar{\signal{#1}}}

\usepackage{ifpdf}
\usepackage{manfnt}
\usepackage{tikz}
\usetikzlibrary{arrows,automata,shapes.geometric, positioning}
\usepackage[mathscr]{eucal}
\usepackage{graphicx}
\graphicspath{ {./fig/} }
\usepackage{amsthm}

\usepackage{cite}

\usepackage{caption}

\usepackage{empheq}
\usepackage{amssymb}

\usepackage{enumerate}

\usepackage{bm} 

\usepackage{gensymb}

\usepackage{amsthm}
\usepackage{amsmath}
\DeclareMathOperator*\argmin{arg \, min}		

\newtheorem{proposition}{Proposition}

\theoremstyle{newremark}

\theoremstyle{newremark}

\theoremstyle{newremark}

\theoremstyle{definition}

\newcommand{\signal}[1]{{\boldsymbol{#1}}}

\newcommand{\real}{{\mathbb R}}
\newcommand{\complex}{{\mathbb C}}
\newcommand{\innerprod}[2]{\left\langle{#1},{#2}\right\rangle}

\newtheorem{remark}{Remark}

\newtheorem{assumption}{Assumption}

\newcommand{\Natural}{{\mathbb N}}
\renewcommand{\refeq}[1]{(\ref{#1})}
\newcommand{\prox}[1]{{\mathrm{prox}_{#1}}}

\newcommand{\dprime}{{\prime\prime}}

\begin{document}
\title{Hybrid data and model driven algorithms for angular power spectrum estimation}

\author{\IEEEauthorblockN{Renato L. G. Cavalcante and Slawomir Sta\'nczak}
\IEEEauthorblockA{\textit{Fraunhofer Heinrich Hertz Institute and Technical University of Berlin, Germany}\\
\{renato.cavalcante, slawomir.stanczak\}@hhi.fraunhofer.de}
}

\maketitle

\begin{abstract}
	We propose two algorithms that use both models and datasets to estimate angular power spectra from channel covariance matrices in massive MIMO systems. The first algorithm is an iterative fixed-point  method that solves a hierarchical problem. It uses model knowledge to narrow down candidate angular power spectra to a set that is consistent with a measured covariance matrix. Then, from this set, the algorithm selects the angular power spectrum with minimum distance to its expected value with respect to a Hilbertian metric learned from data. The second algorithm solves an alternative optimization problem with a single application of a solver for nonnegative least squares programs. By fusing information obtained from datasets and models, both algorithms can outperform existing approaches based on models, and they are also robust against environmental changes and small datasets.
\end{abstract}


\section{Introduction}
Estimating the angular power spectrum (APS) of a signal impinging on an antenna array from the measured channel covariance matrix is an ill-posed problem with important applications in massive MIMO systems, including pilot decontamination \cite{renato2020}, channel covariance matrix estimation in frequency division duplex (FDD) systems \cite{miretti18,miretti18SPAWC,renato18error,dec2015,dec2016,khalil18}, and localization \cite{dec18}, among others. Current approaches for APS estimation can be divided into two main groups: model based methods \cite{miretti18,miretti18SPAWC,renato18error,khalil18} and data driven methods \cite{song2019}.

Model-based methods are able to produce reliable estimates with little side information, no training, and potentially low computational complexity \cite{miretti18,miretti18SPAWC,renato18error,renato2020}. However, they do not exploit any information from datasets to improve the estimates or to gain robustness against measurement errors or model uncertainty, or both. In contrast, pure data-driven methods can provide good performance without any knowledge about physical models, but their robustness against changes in the propagation environment (i.e., the distribution of the APS) is not acceptable for many applications. Furthermore, even if the environment does not change, in general these methods are heuristics that do not provide any guarantees that the APS estimates are consistent with measured covariance matrices. In other words, using an APS estimate in the forward problem that computes the covariance matrix from the APS may not reproduce the measured covariance matrix accurately, and we note that this type of consistency is important to bound errors in some applications, such as the error of channel covariance matrix conversion in FDD massive MIMO systems  \cite{renato18error,hag2018multi}.

Against this background, we propose algorithms that use datasets to improve the estimates obtained with model-based methods, without unduly losing robustness against environmental changes. To this end, we start by revisiting existing algorithms for APS estimation to establish their equivalence and to understand their limitations. In particular, using common assumptions in the literature, we prove that some of these algorithms solve equivalent optimization problems (Proposition~\ref{prop.l1}), in the sense that the set of solutions is the same. However, this set is not a singleton in general, so the performance of these existing algorithms can differ significantly because they may converge to different solutions. Nevertheless, we show in this study that nonuniqueness of the solution can be exploited with the paradigm of hierarchical optimization \cite{yamada11,yamada19} to improve the quality of the estimates. More precisely, from the set of solutions to the existing problem formulations, we select an estimate that least deviates from the expected value with respect to a Hilbertian metric learned from datasets; namely, the Mahalanobis distance. The unique solution to the resulting problem is then reinterpreted as a projection onto the set of fixed points of a proximal mapping, and it is computed via Haugazeau's algorithm \cite[Ch.~30]{baus11}. As an alternative to this iterative method, we also pose an optimization problem that can be solved with a single application of a solver for nonnegative least squares problems. Simulations show that the proposed techniques outperform previous algorithms in some scenarios, and they can be made robust against changes of the distribution of the APS, which is one of the major limitations of data driven methods, and, in particular, neural networks.

\section{Preliminaries}
\label{sec:math}

Hereafter, by $(\cdot)^t$, $(\cdot)^H$, and $(\cdot)^\dagger$ we denote, respectively, the transpose, the Hermitian transpose, and the pseudo-inverse. The set of nonnegative reals is $\real_+$. The real and imaginary components of a complex matrix $\signal{M}\in\mathbb{C}^{N\times N}$ are given by, respectively, $\mathrm{Re}(\signal{M})\in\real^{N\times N}$ and $\mathrm{Im}(\signal{M})\in\real^{N\times N}$.

By $(\mathcal{H},\innerprod{\cdot}{\cdot}_\mathcal{H})$ we denote a real Hilbert space with the inner product $\innerprod{\cdot}{\cdot}_\mathcal{H}$ and induced norm $\|x\|_\mathcal{H}:=\sqrt{\innerprod{x}{x}}_\mathcal{H}$. The set of lower semicontinuous convex functions $f:\mathcal{H}\to\real\cup\{\infty\}$ is given by $\Gamma_0(\mathcal{H})$. The proximal mapping \linebreak[4] $\mathrm{prox}_{ f}:\mathcal{H}\to\mathcal{H}$ of  $f\in \Gamma_0(\mathcal{H})$ maps $\signal{x}\in\mathcal{H}$ to the unique solution to: $\mathrm{Minimize}_{\signal{y}\in\mathcal{H}} f(\signal{y})+(1/2)\|\signal{x}-\signal{y}\|^2_\mathcal{H} $. A function $f:\mathcal{H}\to\real\cup\{\infty\}$ is said to be coercive if $\|\signal{x}\|_\mathcal{H}\to\infty$ implies $f(\signal{x})\to\infty$. The projection $P_C:\mathcal{H}\to C$ onto a nonempty closed convex set $C\subset\mathcal{H}$ maps $\signal{x}\in\mathcal{H}$ to the unique solution to: $\mathrm{Minimize}_{\signal{y}\in C}\|\signal{x}-\signal{y}\|_\mathcal{H}$. The indicator of a set $C\subset\mathcal{H}$ is the function ${\iota}_C:\mathcal{H}\to\{0,\infty\}$ given by $\iota_C(\signal{x}) = 0$ if $\signal{x}\in C$ or $\iota_C(\signal{x}) = \infty$ otherwise. The norms $\|\cdot\|_1$ and $\|\cdot\|_2$ are, respectively, the standard $l_1$ and $l_2$ norms in Euclidean spaces. The set of fixed points of a mapping $T:\mathcal{H}\to\mathcal{H}$ is denoted by $\mathrm{Fix}(T):=\{x\in\mathcal{H}~|~T(x)=x \}$. Given two real Hilbert spaces $(\mathcal{H}^\prime, \innerprod{\cdot}{\cdot}_{\mathcal{H}^\prime})$ and $(\mathcal{H}^\dprime,\innerprod{\cdot}{\cdot}_{\mathcal{H}^\dprime})$, the set $\mathcal{B}(\mathcal{H}^\prime, \mathcal{H}^\dprime)$ is the set of bounded linear operators mapping vectors in $\mathcal{H}^\prime$ to vectors in $\mathcal{H}^\dprime$.

\section{System Model}
\label{sect:model}

We consider the uplink of a system with one single-antenna user and one base station equipped with $N\in\Natural$ antennas. At time $k\in\Natural$, the signal received at the base station spaced by multiples of the coherence interval $T_c$ in a memoryless flat fading channel is given by 
$\signal{y}[k] = \signal{h}[k]~s[k] + \signal{n}[k] \in \mathbb{C}^N,$
where $s[k]\in\mathbb{C}$ and $\signal{h}[k]\in\mathbb{C}^N$ denote, respectively, the transmitted symbol and the channel of the user; and $\signal{n}[k]\in\mathbb{C}^N$ is a sample from the distribution $\mathcal{N}_\mathbb{C}(\signal{0},\sigma^2\signal{I})$. As common in the literature \cite{ad2013,dec2015}, we assume that $E[|s[k]|^2] = 1$ and $E[s[k]]=0$ for every $k\in\Natural$.~\footnote{We use the same notation for random variables and their samples. The meaning that should be applied is clear from the context.} Furthermore, the transmitted symbols and noise are mutually independent, and their distributions do not change with the index $k$ in a sufficiently large time window. Therefore, hereafter we assume that
\begin{align}
\label{eq.contamination}
(\forall k\in\Natural)~E[\signal{y}[k]\signal{y}[k]^H] = \signal{R} + \sigma^2 \signal{I},
\end{align}
where $\signal{R} = E[\signal{h}[k] \signal{h}[k]^H] = \signal{USU}^H \in\complex^{N\times N}$ is the channel covariance matrix, $\signal{U}\in \complex^{N\times N}$ is the unitary matrix of eigenvectors of $\signal{R}$, and $\signal{S}\in \complex^{N\times N}$ is the diagonal matrix of eigenvalues of $\signal{R}$. The channel sample $\signal{h}[k]$ at time $k\in\Natural$ takes the form $\signal{h}[k]=\signal{U} \signal{S}^{1/2}\signal{w}[k]$, where $(\signal{w}[k])_{k\in\Natural}\subset\mathbb{C}^N$ are samples of i.i.d. random vectors with distribution $\mathcal{N}_\mathbb{C}(\signal{0},\signal{I})$. Hereafter, since the distribution of the random variables do not change with the time index $k$ in the memoryless channel described above, we omit this index if confusion does not arise. 

\section{The estimation problem}
\label{sect.relations}

Let $\left(\ha, \innerprod{\cdot}{\cdot}_{\ha}\right)$ be the real Hilbert space of (equivalent classes of) real square integrable functions $\ha=L_2(\Omega)$ with respect to the standard Lebesgue measure $\mu$ on a nonnull measurable set $\Omega\subset \real^M$. In this Hilbert space, inner products are defined by $(\forall x\in\ha)(\forall y\in\ha)\innerprod{x}{y}_\ha=\int_{\Omega} x~y~ \mathrm{d}\mu$. Now, suppose that an array with $N\in\Natural$ antennas at a base station scans signals arriving from angles within a compact domain $\Omega\subset \real^M$, where each coordinate of $\Omega$ corresponds to azimuth or elevation angles, possibly by also considering different antenna polarizations \cite{miretti18SPAWC}. Given ${\theta}\in\Omega$, we denote by $\rho(\theta)$ the average angular power density impinging on the array from angle $\theta$, and we further assume that the function $\rho:\Omega\to\real$, hereafter called the \emph{angular power spectrum} (APS), is an element of $\ha$; i.e., $\rho\in\ha$.  Being a power spectrum, $\rho$ is also an element of the cone 
\begin{align}
\label{eq.cone}
\mathcal{K}:=\{\rho\in\ha~|~\mu(\{\theta\in\Omega~|~\rho(\theta)< 0\}) = 0 \}
\end{align}
of $\mu$-almost everywhere (a.e.) nonnegative functions.

As shown in \cite{miretti18,miretti18SPAWC,renato18error,renato2020}, a common feature of realistic massive MIMO models is that the stacked version  $$\signal{r} = [r_1,\ldots,r_{2N^2}]^t = \phi(\signal{R})$$ of the channel covariance matrix $\signal{R}$ in \refeq{eq.contamination} is related to the angular power spectrum $\rho$ by
\begin{align}
\label{eq.innerprod}
(\forall n\in\{1,\ldots,2N^2\})~r_n = \innerprod{\rho}{g_n}_\ha,
\end{align}
where  $(g_n)_{n\in\{1,\ldots,2N^2\}}$ are functions in $\ha$ defined by physical properties of the array and the propagation model, and $$\phi:\complex^{N\times N}\to\real^{2N^2}:\signal{R}\mapsto \mathrm{vec}\left(\left[\begin{matrix}\mathrm{Re}(\signal{R}) \\ \mathrm{Im}(\signal{R})\end{matrix}\right]\right)$$ is the bijective mapping that vectorizes the imaginary and real components of a matrix. Therefore, in light of \refeq{eq.innerprod}, if the Hilbert space $(\hb, \innerprod{\cdot}{\cdot}_{\hb})$ denotes the standard Euclidean space $\hb = \real^{2N^2}$ equipped with inner product $$(\forall \signal{y}\in\hb)(\forall \signal{x}\in\hb) \innerprod{\signal{x}}{\signal{y}}_{\hb}:=\signal{x}^t\signal{y},$$ then the relation between $\rho$ and $\signal{r}$  is given by $\signal{r}=T\rho$, where 
$T\in\mathcal{B}(\ha, \hb)$ is the operator  \cite{renato2020} 
\begin{align}
\label{eq.lin_op}
\begin{array}{rcl}
T:\ha&\to&\hb \\ 
\rho&\mapsto& [\innerprod{\rho}{g_1}_\ha,\ldots, \innerprod{\rho}{g_{2N^2}}_\ha]^t.
\end{array}
\end{align}

\begin{remark}
	Covariance matrices $\signal{R}$ have structure, so we can remove many redundant equations in \refeq{eq.innerprod} to reduce the dimensionality of the space $\hb$. 
\end{remark}

The objective of the algorithms we propose in this study is to \emph{estimate $\rho$ from a known (vectorized) channel covariance matrix $\signal{r}:=\phi(\signal{R})=T\rho$.} Note that the operator $T$ does not have an inverse in general, so this estimation problem is ill-posed. In particular, the null space $\mathcal{N}(T):=\{x\in\ha~|~ Tx = 0 \}$ of $T\in\mathcal{B}(\ha,\hb)$ is nontrivial (i.e., $\mathcal{N}(T)\neq \{0\}$), so there exist uncountably many functions $\rho$ in $\ha$ for which $T$ maps $\rho$ to the same vector $\signal{r}=\phi(\signal{R})$. Nevertheless, the studies in \cite{miretti18,miretti18SPAWC,renato18error,khalil18} have shown that good estimates of $\rho$ can be obtained with computationally efficient methods in practice. To improve upon these existing methods, we first need to understand their strengths and limitations, which is the topic of the next section. Before we proceed, we discretized all signals and operators to avoid unnecessary technical digressions. However, we emphasize that the results in this study can be straightforwardly extended to the infinite dimensional case described above with the tools in \cite{miretti18,miretti18SPAWC,renato2020,renato18error}. 

To obtain a finite dimension approximation of the estimation problem, we denote by $\rd:=[\rho(\theta_1),\ldots,\rho(\theta_D)]^t \in \real^D$ the discrete version of true angular power spectrum $\rho\in\ha$, where $D$ is size of the discrete grid.\footnote{This approximation is somewhat heuristic because $\ha$ is an equivalence class of functions. In particular, given $\theta\in\Omega$ and $\rho\in\ha$, the value $\rho(\theta)$ is not well defined.} As a result, the integrals in \refeq{eq.innerprod} can be approximated by $(\forall n\in\{1,\ldots,2N^2\})$
\begin{align*}
\innerprod{\rho}{g_n}_\ha = \int_{\Omega} \rho~g_n \mathrm{d}\mu \approx \rd^t \gn,
\end{align*}
where $\gn=(\mu(\Omega)/D) [g_n(\theta_1),\ldots,g_n(\theta_D)]^t\in\real^D$ is a discrete approximation of the function $g_n$ of array. In turn, with \begin{align}
\label{eq.mat_a}
\signal{A}:=[\signal{g}_{\mathrm{d},1}\ldots\signal{g}_{\mathrm{d},2N^2}]^t,
\end{align} 
the operator  
$T_\mathrm{d}:\real^D\to \real^{2N^2}:\signal{\rho}\mapsto \signal{A}\signal{\rho}$ is a discrete approximation of $T$ in \refeq{eq.lin_op}, 
and $\kd:=\real_+^D$ is a discrete approximation of $\mathcal{K}$.

\section{Existing solutions for angular power spectrum estimation}
\label{sect.relations}

For simplicity, in this section we use the following assumption, which is dropped later in Sect.~\ref{sect.proposed}. 

\begin{assumption}
	\label{assumption.range}
	The estimated covariance matrix $\signal{R}$, or, equivalently, $\signal{r}=\phi(\signal{R})$, is compatible with the array, in the sense that it can be generated with one function in $\kd$; i.e., $\phi(\signal{R})=\signal{r}\in T_\mathrm{d}(\kd):=\{\signal{A\rho} \in \real^{2N^2}~|~ \signal{\rho}\in\kd \}.$
\end{assumption} 

If Assumption~\ref{assumption.range} holds, we can estimate $\rd$ from $\signal{r}$ by solving the following set-theoretic estimation problem, which is a discrete version of one of the infinite dimensional problems posed in \cite{miretti18,miretti18SPAWC,renato18error}:
\begin{align}
\label{eq.feasibility}
\text{Find }\signal{\rho}\in\real^D \text{ such that } \signal{\rho}\in \vd\cap \kd,
\end{align}
where $\vd\subset \real^D$ is the linear variety $\vd := \{\signal{\rho} \in\real^D~|~\signal{A\rho}=\signal{r}\}$; i.e., the set of all (not necessarily nonnegative) vectors that produce the observed channel covariance matrix $\signal{R}=\phi^{-1}(\signal{r})$. The idea of Problem~\refeq{eq.feasibility} is to find an estimate that is consistent with all known information about $\rd$, or, more precisely, with the fact that $\rd$ is nonnegative (i.e., $\rd\in\kd$) and it produced the observed channel covariance matrix $\signal{R}=\phi^{-1}(\signal{r})$ (i.e., $\rd\in \vd$). 
 In this set-theoretic paradigm, any two estimates belonging to both $\vd$ and $\kd$ are equally good because no other information about $\rd$ is assumed to be available.
 
Clearly, a necessary and sufficient condition for the convex feasibility problem in \refeq{eq.feasibility} to have a solution is that Assumption~\ref{assumption.range} holds. Since the projections onto $\vd$ and $\kd$ are easy to compute in the Hilbert space $(\mathcal{H}_2,\innerprod{\cdot}{\cdot}_{\mathcal{H}_2})$ \cite[Ch.~3]{stark98}, a plethora of simple iterative projection-based algorithms with convergence guarantees are widely available \cite{stark98,baus16,censor11}. In particular, the variant of the Douglas-Rachford splitting method studied in \cite{baus16} converges in a finite number of iterations. We can also reformulate Problem \refeq{eq.feasibility} as a standard convex program to enable us to use traditional solvers. For example, consider the problem below, which has been proposed in \cite{khalil18}:
\begin{align}
\label{eq.nnls}
\text{Minimize}_{\signal{\rho}\in \kd} \|\signal{A\rho}-\signal{r}\|_2^2.
\end{align}

From the definition of the linear variety $\vd$, any estimate $\signal{\rho}\in \vd$ satisfies $\|\signal{A\rho}-\signal{r}\|_2^2=0$, which is the global minimum of the cost function in Problem \refeq{eq.nnls}. Therefore, under  Assumption~\ref{assumption.range}, we verify that $\signal{\rho}^\star$ solves Problem \refeq{eq.feasibility} if and only if $\signal{\rho}^\star$ solves Problem \refeq{eq.nnls}. We emphasize that Problems \refeq{eq.feasibility} and \refeq{eq.nnls} do not have a unique solution in general. As a result, the quality of the estimate of $\signal{\rho}_\mathrm{d}$ obtained by solving either \refeq{eq.feasibility} or \refeq{eq.nnls} depends on the choice of the iterative solver. 

Nonuniqueness of the solution provides us with additional possibilities to choose a vector in the solution set with additional desirable properties. For example, a common \emph{hypothesis} is that $\signal{\rho}_\mathrm{d}$ is a sparse vector, so, as an attempt to promote sparsity, we may select a solution to \refeq{eq.feasibility} with minimum $l_1$ norm (recall that the $l_1$ norm is known to promote sparsity). Formally, we solve the following problem:
\begin{align}
\label{eq.l1}
\text{Minimize}_{\signal{\rho}\in\real^D} \|\signal{\rho}\|_1 \text{ subject to } \signal{\rho}\in \vd\cap\kd.
\end{align}
However, as we argue below, for common array models in the literature, there is nothing to be gained by solving \refeq{eq.l1} instead of \refeq{eq.feasibility} or the equivalent problem in \refeq{eq.nnls} [if Assumption~\refeq{assumption.range} holds] because the set of solutions to Problems \refeq{eq.feasibility}, \refeq{eq.nnls}, and \refeq{eq.l1} are the same. Some of these arrays satisfy the following assumption:

\begin{assumption}
	\label{a.one}
	Let $S:=\{g_1, \ldots, g_{2N^2}\} \subset \ha$ be the set of functions of the array. We assume that the function $u:\Omega\to\real: \theta\mapsto 1$ is a member of $S$, in which case the vector $c \signal{1}$, where $\signal{1}\in\real^D$ is the vector of ones and $c:=\mu(\Omega)/D$, is a row of the matrix $\signal{A}$ in \refeq{eq.mat_a}.
\end{assumption}

\begin{remark}
	\label{remark.ula}
	Assumption~\ref{a.one} is valid for common array models with isotropic antennas, such as uniform linear arrays and planar arrays.
\end{remark}

The relation among Problems \refeq{eq.feasibility}, \refeq{eq.nnls}, and \refeq{eq.l1} is formally established in the next simple proposition.
\begin{proposition}
	\label{prop.l1}
	Let Assumptions~\ref{assumption.range} and \ref{a.one} be valid. Then set of solutions to Problems \refeq{eq.feasibility}, \refeq{eq.nnls}, and \refeq{eq.l1} are the same.
\end{proposition}
\begin{proof}
	If Assumption~\ref{assumption.range} holds, then $\vd\cap\kd\neq\emptyset$. Now, let $\signal{\rho}\in \vd\cap\kd$ be arbitrary. Assumption~\ref{a.one} implies that, for $c:=\mu(\Omega)/D$, there exists $k\in\{1,\ldots,2N^2\}$ such that $r_k\overset{(\text{a})}{=}  c\signal{1}^t\signal{\rho}\overset{(\text{b})}{=}c~\|\signal{\rho}\|_1$, where (a) follows from $\signal{\rho}\in \vd$ and (b) follows from $\signal{\rho}\in\kd$. Since $\signal{\rho}$ is arbitrary, we conclude that all vectors in $\vd\cap\kd$ have the same $l_1$ norm, which implies that Problems \refeq{eq.feasibility} and \refeq{eq.l1} have the same set of solutions. The equivalence between Problems \refeq{eq.feasibility} and \refeq{eq.nnls} has already been established, so the proof is complete.
\end{proof}

The practical implication of Proposition~\ref{prop.l1} is that Problems~\refeq{eq.feasibility} and \refeq{eq.nnls} are expected to promote sparsity implicitly, but the estimand $\rd$ is not necessarily the sparsest vector of the solution set. Therefore, we need additional information in the problem formulations to improve the estimates,  and in the next section we incorporate statistical information gained from datasets. 

\section{Proposed algorithms}
\label{sect.proposed}
Given \emph{a positive definite matrix} $\signal{M}\in\real^{D\times D}$ [this matrix is fixed later in \refeq{eq.mat_m}], let $(\hmr, \innerprod{\cdot}{\cdot}_{\hmr})$ denote the Hilbert space consisting of the vector space $\hmr:=\real^D$ equipped with the inner product $(\forall \signal{x}\in\hmr)(\forall \signal{y}\in\hmr) \innerprod{\signal{x}}{\signal{y}} = \signal{x}^t\signal{My}$. By definition, the vector space $\hmr=\real^D$ does not depend on $\signal{M}$, but the notation $\mathcal{H}_{\signal{M}}$ is useful to clarify the inner product defined on $\real^D$.

Now, assume that a dataset $\mathcal{M}=\{\signal{\rho}_{\mathrm{d},1},\ldots, \signal{\rho}_{\mathrm{d},L}\}$ with $L$ samples of angular power spectra is available, and suppose that these samples have been independently drawn from the same distribution with mean $\hsignal{\rho}\in\kd$ and covariance matrix $\signal{C}\in\real^{D\times D}$. In practice, $\hsignal{\rho}$ and $\signal{C}$ can be estimated from a sample average as follows (assuming $L\gg 1$):
\begin{align}
\label{eq.rho_mean}
\hsignal{\rho} \approx \dfrac{1}{L} \sum_{n=1}^L  \signal{\rho}_{\mathrm{d},n} \in\real^{D\times D}
\end{align}
and
\begin{align}
\label{eq.C}
\signal{C} \approx \dfrac{1}{L-1} \sum_{n=1}^L  (\signal{\rho}_{\mathrm{d},n}-\hsignal{\rho})(\signal{\rho}_{\mathrm{d},n}-\hsignal{\rho})^t.
\end{align}

Hereafter, to exploit knowledge gained from $\signal{C}$ and $\hsignal{\rho}$, we use the Hilbert space $\left(\hmr,\innerprod{\cdot}{\cdot}_{\hmr}\right)$ defined above by fixing $\signal{M}$ to  
\begin{align} 
\label{eq.mat_m}
\signal{M}_\alpha:=(\signal{C}+\alpha \signal{I})^{-1},
\end{align}
where $\alpha>0$ is a design parameter that serves two purposes: (i) it guarantees positive definiteness of $\signal{M}_\alpha$, and (ii) it provides robustness against environmental changes, as discussed below. An important feature of the Hilbert space $(\hmra,\innerprod{\cdot}{\cdot}_{\hmra})$ is that its induced norm $(\forall \signal{x}\in \hmra) ~ \|\signal{x}\|_{\hmra} := \sqrt{\innerprod{\signal{x}}{\signal{x}}_{\hmra}}$ in turn induces the Hilbertian metric $(\forall \signal{x}\in \hmra)(\forall \signal{y}\in \hmra) ~ d_{\hmra}(\signal{x},\signal{y}):=\|\signal{x}-\signal{y}\|_{\hmra}$ that is known as the Mahalanobis distance in statistical pattern recognition \cite{gjm2004}. In particular, if the design parameter $\alpha>0$ is sufficiently small, the distance $d_{\hmra}(\signal{x},\hsignal{\rho})$ between the distribution mean $\hsignal{\rho}$ and a given vector  $\signal{x}\in\hmra$ is known to provide us with a notion of distance between $\signal{x}$ and the distribution of the dataset $\mathcal{M}$.  As the parameter $\alpha$ increases, the influence of the dataset in the metric $d_{\hmra}$ decreases ($d_{\hmra}$ becomes increasingly similar to a scaled version of the standard Euclidean metric), so large $\alpha$ can be useful in scenarios in which the distribution of the angular power spectrum changes significantly over time and acquisition of datasets is difficult. We now propose two algorithms based on the Hilbert space $\left(\hmra,\innerprod{\cdot}{\cdot}_{\hmra}\right)$.

\subsection{Algorithm 1}
In Sect.~\ref{sect.relations} we have shown that Problems \refeq{eq.feasibility}, \refeq{eq.nnls}, and \refeq{eq.l1} do not have a unique solution in general, and they are equivalent if the assumptions in Proposition~\ref{prop.l1} hold. Therefore, among all  solutions, we propose to select the solution with minimum distance to the distribution of the dataset in the sense defined above; i.e.,  we minimize the Mahalanobis distance. Formally, given $\alpha>0$, we solve the following hierarchical problem:
\begin{align}
\label{eq.generalization}
\text{Minimize}_{{\signal{\rho}}\in S} \|{\signal{\rho}}-\hsignal{\rho}\|_{\hmra},
\end{align}
where
\begin{align}
\label{eq.set_s}
S:=\argmin_{\signal{\rho}\in\hmra} g(\signal{\rho}) \subset \hmra
\end{align}
and 
\begin{align}
\label{eq.func_g}
\Gamma_0(\hmra)\ni g:\hmra\to\real_+:\signal{\rho}\mapsto \|\signal{A\rho}-\signal{r}\|_{2}^2+\iota_{\kd}(\signal{\rho}).
\end{align}
Note that $S$ is the set of solutions to Problem~\refeq{eq.nnls}, and, if the assumptions in Proposition~\ref{prop.l1} hold, then $S$ is also the set of solutions to Problems~\refeq{eq.feasibility} and \refeq{eq.l1}. However, hereafter we do not necessarily assume that the assumptions in Proposition~\ref{prop.l1} hold. In particular, as discussed below, the proposed algorithm can deal with the case $\kd\cap\vd=\emptyset$ without any changes.

One of the challenges for solving \refeq{eq.generalization} is that hierarchical problems are not in general canonical convex programs as defined in some well-known references \cite{boyd}, where constraints have to be expressed as level sets of convex functions or as equalities involving affine functions. Therefore, the solvers described in these references are not directly applicable. The proposed strategy for solving \refeq{eq.generalization} is to interpret its solution as the projection from $\hsignal{\rho}$ onto the fixed point set of a computable firmly nonexpansive mapping, which enables us to apply best approximation techniques such as those based on Haugazeau's algorithm \cite[Theorem~30.8]{baus17}. 

In more detail, recalling the definition of projections, we verify that the solution $\signal{\rho}^\star$ to \refeq{eq.generalization} is the  projection from $\hsignal{\rho}$ onto the closed convex set $S$ in the Hilbert space $(\hmra, \innerprod{\cdot}{\cdot}_{\hmra})$; i.e.,  $\signal{\rho}^\star=P_S(\hsignal{\rho})$. As a result, the solution exists and is unique provided that the set $S$ is nonempty, and we can show nonemptiness of this set even if we weaken the assumptions in Proposition~\ref{prop.l1}. For example, let us only assume that one of the vectors $(\signal{g}_{\mathrm{d},n})_{n\in\Natural}$ has (strictly) positive components (see Assumption~\ref{a.one} and Remark~\ref{remark.ula}). In this case, we can show that $g$ is coercive, but we omit the details for brevity. Therefore, we have $S\neq \emptyset$ as an implication of \cite[Proposition~11.15]{baus17}. 

The projection onto $S$ does not have a closed-form expression in general, but it can be computed with iterative methods. To this end, note that the set $S$ can be equivalently expressed as the fixed point set of the mapping $\mathrm{prox}_{\gamma g}:\hmra\to\hmra$ for every $\gamma>0$; i.e., $(\forall \gamma>0)~\mathrm{Fix}(\mathrm{prox}_{\gamma g}) = S$. Therefore, given an arbitrary scalar $\gamma>0$, the desired solution $\signal{\rho}^\star = P_S(\hsignal{\rho}) = P_{\mathrm{Fix}(\mathrm{prox}_{\gamma g})}(\hsignal{\rho})$ is the limit of the sequence  $(\signal{\rho}_n)_{n\in\Natural}$ constructed with the following instance of Haugazeau's algorithm:
\begin{align}
\label{eq.hag}
\signal{\rho}_{n+1} = Q(\signal{\rho}_1, \signal{\rho}_n, \mathrm{prox}_{\gamma g}(\signal{\rho}_n)),
\end{align}
where $\signal{\rho}_1:=\hsignal{\rho}$, 
\begin{align*}
\begin{array}{l}
Q:\hmra \times \hmra \times \hmra \to \real \\
(\signal{x},\signal{y},\signal{z})\mapsto  \begin{cases}
\signal{z},\text{ if }\delta = 0 \text{ and } \chi\ge 0; \\
\signal{x}+\left(1+\dfrac{\chi}{\nu}\right)(\signal{z}-\signal{y}), \\ \qquad \text{ if }\delta > 0 \text{ and } \chi\nu \ge \delta;\\
\signal{y}+\dfrac{\nu}{\delta}\left(\chi(\signal{x}-\signal{y})+\mu(\signal{z}-\signal{y})\right),\\ \qquad \text{ if }\delta > 0 \text{ and } \chi\nu < \delta;
\end{cases}
\end{array}
\end{align*} 
$\chi = \innerprod{\signal{x}-\signal{y}}{\signal{y}-\signal{z}}_{\hmra}$, $\mu=\|\signal{x}-\signal{y}\|^2_{\hmra}$, $\nu = \|\signal{y}-\signal{z}\|^2_{\hmra}$, and $\delta=\mu\nu-\chi^2$.

The proof that the sequence $(\signal{\rho}_n)_{n\in\Natural}$ constructed via \refeq{eq.hag} indeed converges to $P_S(\hsignal{\rho})$ is a simple application of \cite[Theorem~30.8]{baus17}. More precisely, recall that proximal mappings are firmly nonexpansive, so the mapping $\signal{x}\mapsto \signal{x}-\mathrm{prox}_{\gamma g}(\signal{x})$ is demiclosed everywhere \cite[Theorem~4.27]{baus17}. Therefore, we fulfill all the conditions in \cite[Theorem~30.8]{baus17} for the sequence constructed via \refeq{eq.hag} to converge to $P_{\mathrm{Fix}(\mathrm{prox}_{\gamma g})}(\hsignal{\rho})=P_{S}(\hsignal{\rho})$. 

\begin{remark} \label{remark.proxg} {\it (Computation of the proximal mapping of $g$)}
	Using the definition of proximal mappings, after simple algebraic manipulations, we verify that $\mathrm{prox}_{\gamma g}(\signal{x})$ in the Hilbert space $(\hmra, \innerprod{\cdot}{\cdot}_{\hmra})$ for given $\signal{x}\in\hmra$ and $\gamma>0$ is the solution to
	\begin{align}
	\label{eq.problem_prox}
	\mathrm{Minimize}_{\signal{y}\in \kd}  \|\signal{Q}^{1/2}\signal{y}-\signal{b}\|_2^2,
	\end{align}	 
	where $\signal{Q}^{1/2}$ is the principal square root of $\signal{Q}:=\signal{A}^t\signal{A}+(1/(2\gamma))\signal{M}_\alpha,$ and $\signal{b}:=\signal{Q}^{-1/2}(\signal{A}^t\signal{r}+(1/(2\gamma))\signal{M}_\alpha\signal{x}).$ 
	Problem \refeq{eq.problem_prox} is a standard nonnegative least-squares program, so the proximal mapping $\mathrm{prox}_{\gamma g}:\hmra\to\hmra$ can be computed with solvers that terminate with a finite number of steps, such as those based on the active-set method \cite{law1995}.
\end{remark}

\subsection{Algorithm~2}

To derive a low-complexity alternative to Algorithm~1, we modify Problem~\refeq{eq.nnls} by adding a regularizer based on the Mahalanobis distance as follows:
\begin{align}
\label{eq.alg2}
\text{Minimize}_{\signal{\rho}\in\hmra}  \|{\signal{\rho}}-\hsignal{\rho}\|_{\hmra}^2+ \mu \|\signal{A}\signal{\rho}-\signal{r}\|_{2}^2 + \iota_{\kd}(\signal{\rho}),
\end{align}
where $\mu>0$ is a design parameter that trades deviations from the set $\vd$ against the distance to the distribution of the dataset, and $\alpha>0$ is the design parameter of the Hilbert space $(\hmra, \innerprod{\cdot}{\cdot}_{\hmra})$. The definition of proximal mappings shows that the unique solution $\signal{\rho}^\star$ to Problem~\refeq{eq.alg2} is $\signal{\rho}^\star=\prox{(\mu/2)g}(\hsignal{\rho})$, where $g$ is the function defined in \refeq{eq.func_g}. As a result, in light of Remark~\ref{remark.proxg}, Problem~\refeq{eq.alg2} can be solved with a single application of the active-set method \cite{law1995}, unlike the algorithm in \refeq{eq.hag}, which uses a nonnegative least squares solver to compute the proximal mapping of $\gamma g$ at each iteration. The price we pay for this reduction in computational effort is that the formulation in \refeq{eq.alg2} requires knowledge of a good value for $\mu$ because the solution to \refeq{eq.alg2} depends on this parameter. In contrast, the parameter $\gamma$ in \refeq{eq.hag} determines the path taken by the iterates, but not the vector to which the algorithm converges. 

\begin{remark}
 Additional regularizers, such as those based on total variation techniques could also be added to \refeq{eq.alg2}, but we do not consider them here because of the space limitation. 
\end{remark}

\section{Simulations and conclusions}
We assume that a base station is equipped with an uniform linear array operating with $N=16$ antennas, frequency $f=2.11$ GHz, speed of wave propagation $c=3\cdot 10^8$ m/s, antenna spacing $d=c/(2f)$, and the array response shown in \cite[Example 1]{renato2020}.  The samples of angular power spectra use a conventional model in the literature \cite{miretti18,renato2020}. More precisely, each run of the simulation constructs an angular power spectrum via  $\rho:\Omega\to\real_+:\theta\mapsto\sum_{k=1}^Q \alpha_k h_k(\theta)$, where $\Omega:=[-\pi/2,~\pi/2]$, $Q$ is uniformly drawn from $\{1,2,3,4,5\}$; $h_k:\Omega\to\real_+:\theta\mapsto ({1}/{\sqrt{2\pi\Delta_k^2}})\exp\left({-{(\theta-\phi_k)^2}/{(2\Delta_k^2)}}\right)$;
$\phi_k$, the main arriving angle of the $k$th path, is uniformly drawn from $[0, \pi/2]$; and $\alpha_k$ is uniformly drawn from $[0,~1]$, and it is further normalized to satisfy $\sum_{k=1}^Q\alpha_k=1$. The discrete grid to approximate angular power spectra has $D=180$ uniformly spaced points. Estimates of channel covariance matrices are produced via $P_\mathcal{T}(\sum_{i=k}^{500} \signal{h}[k] \signal{h}[k]^T - \sigma^2 \signal{I} )$, where $\sigma^2=0.1$ is the noise variance in \refeq{eq.contamination}, and $P_\mathcal{T}:\complex^{N\times N}\to \mathcal{T}$ denotes the projection onto the set $\mathcal{T}\subset \complex^{N\times N}$ of Toeplitz matrices with respect to the complex Hilbert space $(\complex^{N\times N}, \innerprod{\signal{A}}{\signal{B}}=\signal{B}^H\signal{A})$. For the construction of the operator $T$ in \refeq{eq.lin_op}, we use only $2N-1$ functions because channel covariance matrices of uniform linear arrays are Toeplitz. 

The approximations in \refeq{eq.rho_mean} and \refeq{eq.C} use 1,000 samples of angular power spectra, and the parameter $\alpha$ to construct the matrix $\signal{M}_\alpha$ in \refeq{eq.mat_m} is set to $\alpha = \|\signal{C}\|_2/100$, where $\|\signal{C}\|_2$ denotes the spectral norm of the empirical covariance matrix $\signal{C}$ in \refeq{eq.C}. Subsequently, we normalize the matrix $\signal{M}_\alpha$ to satisfy $\|\signal{M}_\alpha\|_2=1$. With an abuse of notation, we use the normalized mean square error (MSE) $E\left[{\|{\signal{\rho}} - \rd\|_2^2}/{\|\rd\|_2^2}\right]$ as the figure of merit to compare different algorithms, where $\signal{\rho}$ is the estimate of  $\signal{\rho}_\mathrm{d}$, and expectations are approximated with the empirical average of 200 runs of the simulation. 

Fig.~\ref{fig.results1} shows the performance of the following algorithms: (i) the extrapolated and accelerated projection method (EAPM) used in \cite{miretti18,miretti18SPAWC} operating in the standard Euclidean space $\real^D$ with inner product $(\forall \signal{x}\in\real^D)(\forall\signal{y}\in\real^D) \innerprod{\signal{x}}{\signal{y}}:=\signal{x}^t\signal{y}$; (ii) Haugazeau's algorithm in \refeq{eq.hag} with $\gamma=5$; and (iii) the solution to the nonnegative least squares (NNLS) problem in \refeq{eq.alg2}, computed with SciPy NNLS solver, with $\mu=5\cdot 10^4$ (NNLS-1) and $\mu=1$ (NNLS-2). Note that the algorithms NNLS-1 and NNLS-2 are not considered iterative methods because we assume that solvers for NNLS programs are available as a computational tool. Therefore, we use the convention that the estimates produced by these algorithms are the same at every iteration. 

We have also simulated a neural network similar to that in \cite[Fig.~5]{song2019} with two modifications. First, the number of neurons in each layer was scaled by 180/128 to account for the finer grid used in this study. Second, the last layer based on the soft-max activation function was replaced by the  rectified linear unit activation function because the desired estimand is nonnegative and the soft-max function is inappropriate for the figure of merit considered above (with the softmax activation function, simply scaling the input deteriorates the performance severely if no additional heuristics are  employed). By carefully training this neural network with different solvers, step sizes, epochs, batch sizes, and with a training set containing 110,200 samples (which is two orders of magnitude larger than the dataset used by the proposed algorithms), we have not obtained a normalized MSE better than $6\cdot 10^{-2}$, which is worse than the MSE obtained with the existing EAPM algorithm in Fig.~\ref{fig.results1}. Furthermore, with the scenario considered later in Figs.~\ref{fig.results2} and \ref{fig.results3} (which uses training and test sets constructed with different distributions), the MSE increases drastically (MSE $> 2$). For these reasons, we do not show the performance of the neural network in the figures.

Some conclusions for this first experiment are as follows:

- The proposed algorithms can outperform the EAPM algorithm used in \cite{miretti18,miretti18SPAWC} because statistical information obtained from a dataset is exploited, and we note that the EAPM algorithm has already been shown to outperform existing data driven methods that can cope with small datasets \cite{miretti18}.

- The performance gap between NNLS-1 and NNLS-2 shows that the solution to Problem~\refeq{eq.alg2} is sensitive to the choice of the regularization parameter $\mu$. Nevertheless, if a good value is known, which can be obtained with cross-validation techniques, then the solution to Problem~\refeq{eq.alg2} has performance similar to that obtained with Haugazeau's method.

\begin{figure}
	\begin{center}
		\includegraphics[width=.9\columnwidth]{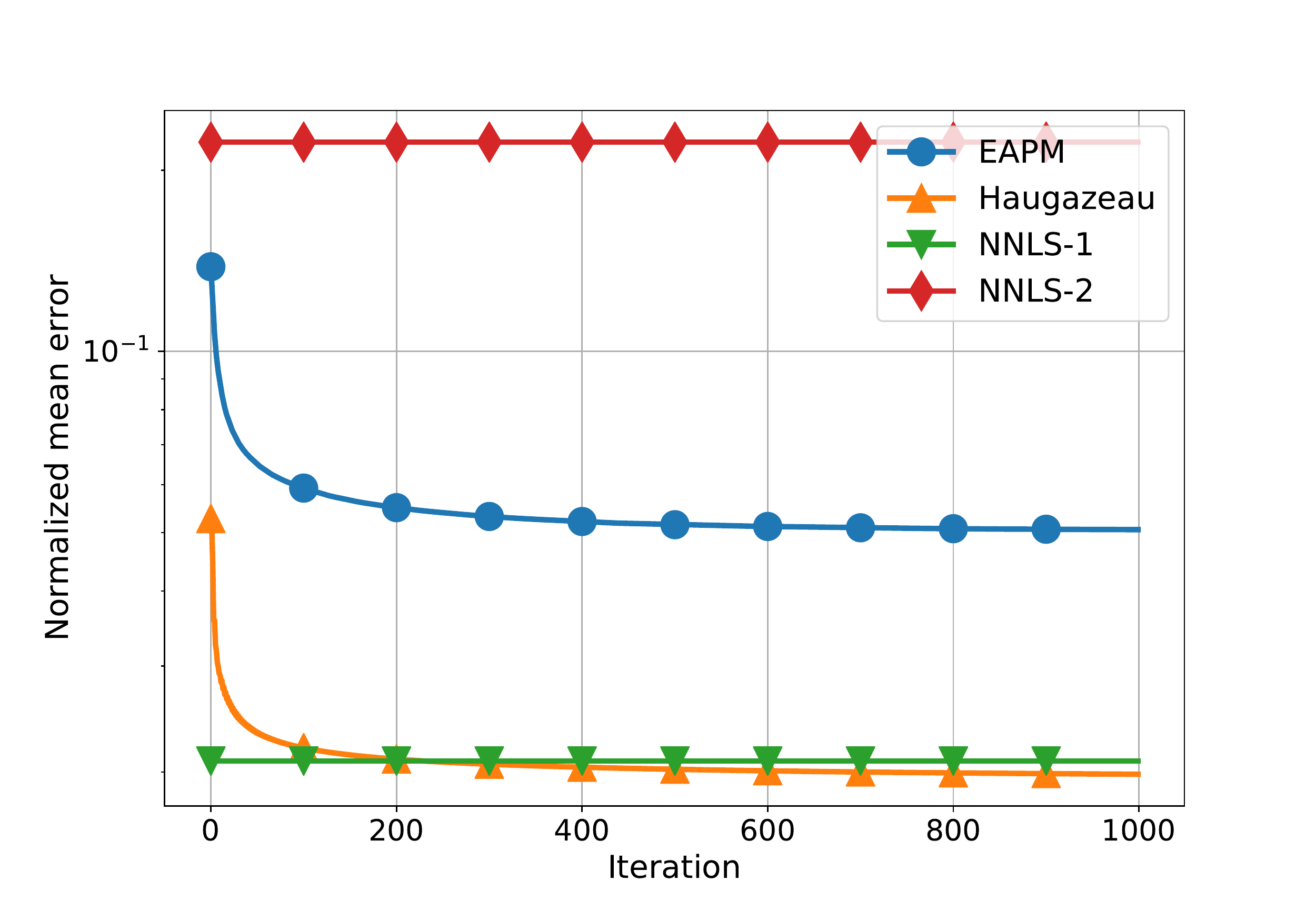}
		\caption{Normalized mean square error as a function of the number of iterations.} 
		\label{fig.results1}
	\end{center}

\end{figure}

A well-known limitation of data-driven methods (and, in particular, neural networks, as discussed above) is the poor generalization performance if the estimand is sampled from a distribution different from that used to construct training sets. As we now show, the proposed hybrid data and model driven algorithms can mitigate problems of this type. 

In Fig.~\ref{fig.results2}, we use the proposed algorithms with the dataset in Fig.~\ref{fig.results1} to reconstruct angular power spectra with the main angles of the paths drawn uniformly at random within the interval $[-\pi/2, 0]$. By doing so, we mimic an extreme scenario where the principal subspaces obtained from the dataset contain almost no energy of the angular power spectra being estimated. As seen in Fig.~\ref{fig.results2}, the performance of the Haugazeau and NNLS-1 hybrid methods  deteriorates, but the MSE does not increase to a point to render these algorithms ineffective. The reason is that the estimates produced by these two proposed algorithms are consistent with the measurements and the array model (i.e., they are close to the set $\vd$), and this fact alone may be enough to provide performance guarantees in some applications, as proved in \cite{renato18error,hag2018multi}. Furthermore, the proposed algorithms have a tunable parameter to make them robust against changes in the distribution of the angular power spectrum; namely, the parameter $\alpha$ in \refeq{eq.mat_m}. This fact is illustrated in Fig.~\ref{fig.results3}, where we show the performance of the algorithms with the parameter $\alpha$  increased to $\alpha=\|\signal{C}\|_2$ (the remaining simulation parameters are the same as those used to produce Fig.~\ref{fig.results2}). The performance of the Haugazeau and NNLS-1 algorithms in Fig.~\ref{fig.results3} approaches the performance of the pure model-based EAPM because, by increasing $\alpha$, the proposed algorithms increasingly ignore the erroneous information about the distribution of the estimand, which is inferred from the dataset.

\begin{figure}
	\begin{center}
		\includegraphics[width=.9\columnwidth]{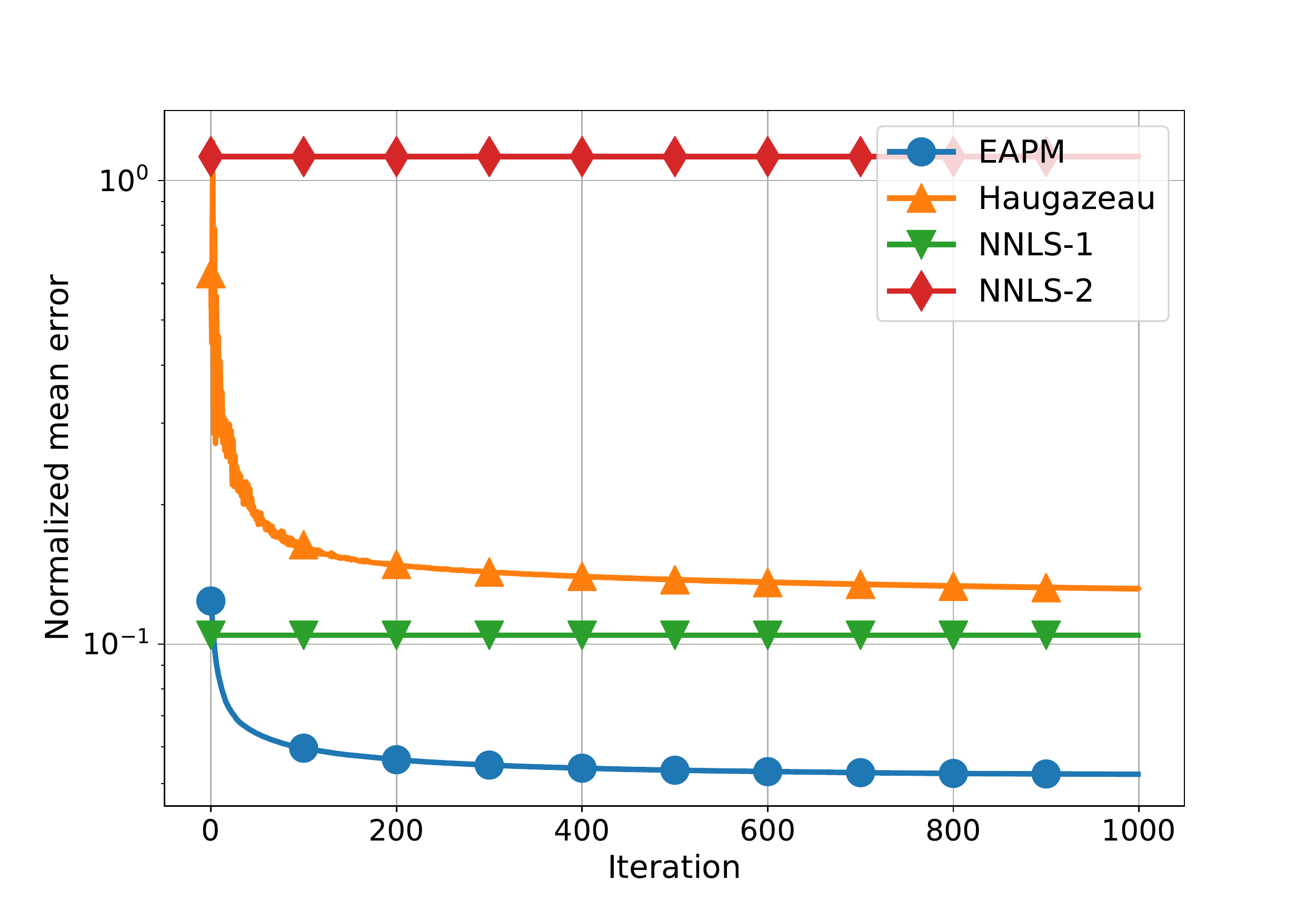}
		\caption{Normalized mean square error as a function of the number of iterations. Angular power spectra of the dataset drawn from a distribution different from that of the estimand ($\alpha=\|\signal{C}\|_2/100$).}
		\label{fig.results2}
	\end{center}
	
\end{figure}

\begin{figure}
	\begin{center}
		\includegraphics[width=.9\columnwidth]{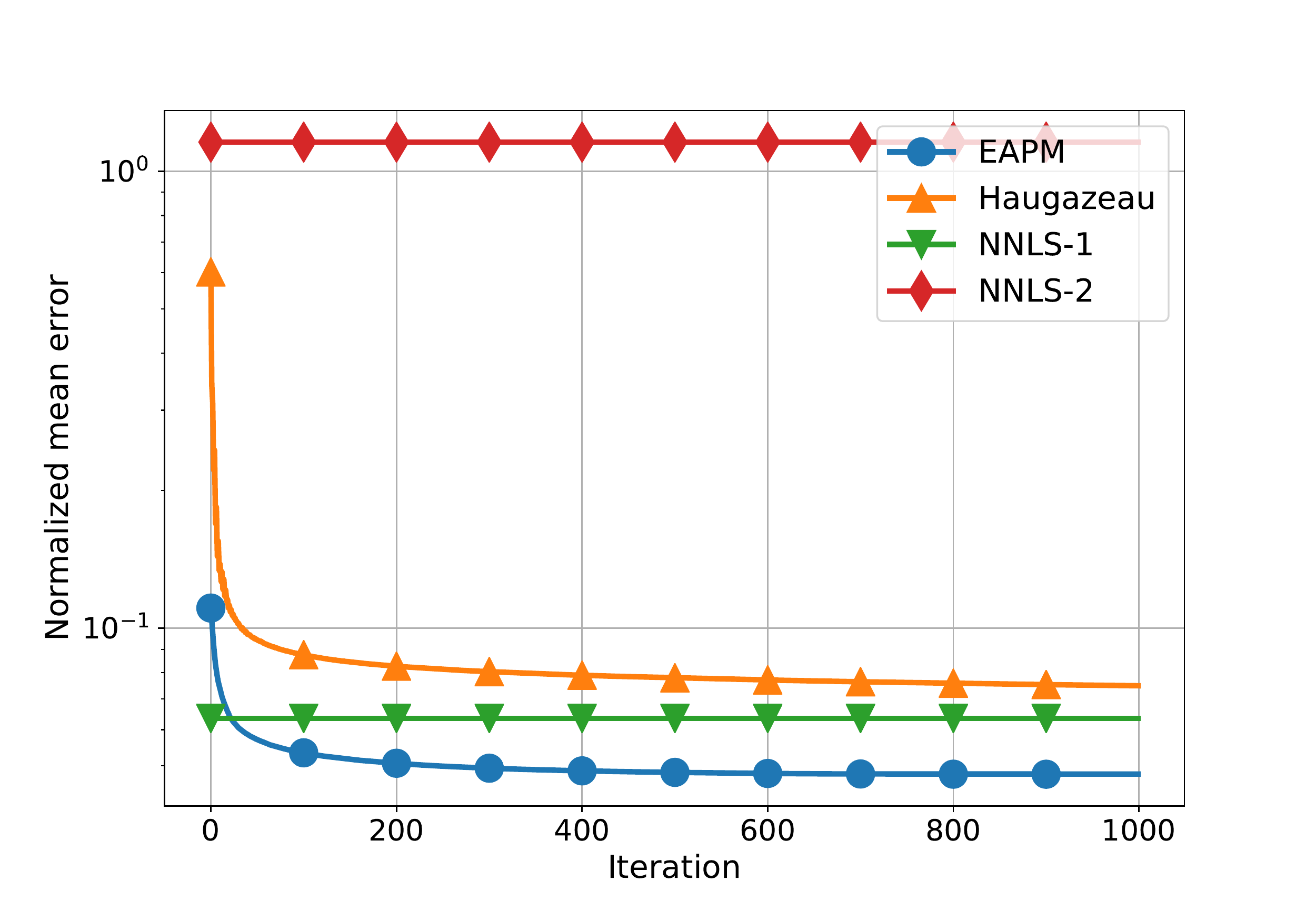}
		\caption{Normalized mean square error as a function of the number of iterations. Angular power spectra of the dataset drawn from a distribution different from that of the estimand ($\alpha=\|\signal{C}\|_2$).}
		\label{fig.results3}
	\end{center}
	
\end{figure}

\bibliographystyle{IEEEtran}
\bibliography{IEEEabrv,references_renato}

\end{document}